\documentclass[11pt,a4paper]{article}
\usepackage[utf8]{inputenc}
\usepackage{graphicx}
\usepackage{mathpazo}

\usepackage{amssymb,amsmath,amsthm}
\usepackage{epsfig}
\usepackage{url}
\usepackage{subfig}
  \usepackage{nicematrix}
\usepackage{xcolor}         
\newtheorem{theorem}{Theorem}[section]
\newtheorem{lemma}[theorem]{Lemma}
\newtheorem{corollary}[theorem]{Corollary}

\newtheorem{remark}[theorem]{Remark}

\newtheorem{claim}[theorem]{Claim}
\usepackage{enumitem}
  \usepackage{tikz}
    \usetikzlibrary{arrows,shadows,positioning}

\usepackage[pdfencoding=auto]{hyperref}
	\usepackage{hyperref}
\hypersetup{colorlinks=true,linkcolor=black,citecolor=[rgb]{0.5,0,0}}

\renewcommand{\tilde}{\widetilde}

\usepackage{complexity}

\usepackage{nicefrac}
\usepackage{setspace}
\usepackage{xspace}

\usepackage[margin=1in]{geometry}

\title{\textbf{On connections between $k$-coloring and Euclidean $k$-means}}

\author{Enver Aman\vspace{0.1cm}\\ \url{enver.aman@rutgers.edu}\vspace{0.1cm}\\ Rutgers University \and Karthik C.\ S.\footnote{This work was  supported by the National Science Foundation under Grant CCF-2313372 and by the Simons  Foundation, Grant Number 825876, Awardee Thu D. Nguyen.}\vspace{0.1cm}\\\url{karthik.cs@rutgers.edu}\vspace{0.1cm}\\ Rutgers University\and Sharath Punna\footnote{This work was done as part of the master thesis at Rutgers University.}\vspace{0.1cm}\\\url{sharathcpunna@gmail.com}\vspace{0.1cm} \\ Ansys, Inc.}
\date{}

\begin{document}

\maketitle

\begin{abstract}
In the Euclidean $k$-means problems we are given as input a set of $n$ points in $\mathbb{R}^d$ and the goal is to   find a set of $k$ points $C\subseteq \mathbb{R}^d$, so as
to minimize the sum of the squared Euclidean distances from each point in $P$ to its closest
center in $C$.
In this paper, we formally explore connections between the $k$-coloring problem on graphs and the Euclidean $k$-means problem. Our results are as follows:

\begin{itemize}
    \item For all $k\ge 3$, we provide a simple reduction from the $k$-coloring problem on regular graphs to the Euclidean $k$-means problem. Moreover, our technique extends to enable a reduction from a  structured max-cut problem (which may be considered as a partial 2-coloring problem) to the Euclidean $2$-means problem. Thus, we have a simple and alternate proof of the NP-hardness of Euclidean 2-means problem. \vspace{0.1cm}
    \item In the other direction, we mimic the $O(1.7297^n)$ time algorithm of Williams [TCS'05] for the  max-cut of  problem on $n$ vertices to obtain an algorithm for the Euclidean 2-means problem with the same runtime, improving on the naive exhaustive search running in   $2^n\cdot \poly(n,d)$ time.\vspace{0.1cm}

    \item We prove similar results and connections as above for the Euclidean \textit{$k$-min-sum} problem.
\end{itemize}
\end{abstract}
  
\clearpage

\section{Introduction} 
The $k$-means problem\footnote{Throughout this paper, we consider the $k$-means problem only in the Euclidean space. } is a classic objective for modelling
clustering in a large variety of applications arising in data mining and
machine learning. Given a set of $n$ points $P\subseteq \mathbb{R}^d$,
the goal is to find a set of $k$ points $C\subseteq \mathbb{R}^d$, called \emph{centers}, so as
to minimize the sum of the squared distances from each point in $P$ to its closest
center in $C$.
The algorithmic study of the $k$-means problem arguably
started with the seminal work of Lloyd~\cite{lloyd1982least}. 
Since then, the problem has received a tremendous amount of attention \cite{Berkhin2006,WuKQGYMMNLYZSHS08}.

The $k$-means problem  is 
known to be NP-Hard, even when the points lie in the Euclidean plane (and $k$ is large)~\cite{MahajanNV12}, or even when $k=2$ (and the dimension is large)~\cite{dasgupta2009random}.
 On the positive side, a near linear time approximation scheme exists  when the  dimension is fixed  (and the number of clusters $k$ is arbitrary)~\cite{CohenAddadFS19}, or when the number of clusters is constant (and the dimension $d$ is arbitrary)~\cite{KumarSS10,BhattacharyaJK18}. 
When both $k$ and $d$ are arbitrary, 
several groups of researchers have shown hardness
of approximation results~\cite{AwasthiCKS15,LeeSW17,Cohen-AddadS19,Cohen-AddadSL22}.

\paragraph{Fine-Grained Complexity.}
One of the main research directions in Fine-Grained Complexity is to identify the exact complexity of important hard problems, distinguishing, say, between NP-complete problems where exhaustive search is essentially the best possible algorithm, and those that have improved exponential time algorithms \cite{Williams15,Williams16,Vir18}. 
The importance of high-dimensional Euclidean inputs in statistics and
machine learning applications has led researchers to study the parameterized and fine-grained complexity of the problem. Both the
dimensionality of the input, $d$, and the target number of clusters, $k$, have been studied as parameters, in as early as the
mid 90s. 

The $k$-means problems can be solved in time $O(k^n\cdot \poly(n,d))$ by simply performing an exhaustive search over the solution space.  This can be improved to $4^n\cdot \poly(n,d)$ runtime using dynamic programming \cite{jensen1969dynamic}, which itself can be further improved to $2^n\cdot \poly(n,d)$ runtime using fast max-sum convolution \cite{bjorklund2007fourier}. The seminal work of Inaba, Katoh, and Imai~\cite{InabaKI94} has shown that one can compute an exact solution to the
$k$-means problems in time $n^{kd+1}$. 
However, this algorithm clearly suffers from the so-called ``curse of dimensionality'', the higher the dimension, the
higher the running time. Thus, for high dimensions, say when $d=n$, the algorithm in \cite{InabaKI94} is slower than even the exhaustive search over the solution space. Therefore, we ask:
\begin{center}
{\textit{Can we beat $2^n$ runtime for $k$-means problem in high dimensions?\footnote{ In \cite{fomin2022exact}, the authors provide a better than $2^n$ runtime algorithm (to be precise an $ 1.89^n\poly(n,d)$ runtime algorithm) for a \emph{discrete} variant of the $k$-means problem, where the centers need to be picked from the input point-set. Their result extends to the discrete variant of the $k$-median and $k$-center problems as well. However, the discrete variant is not as natural as the continuous variant in geometric spaces. } }}
\end{center}

The complexity of $k$-means problem increases as the number of cluster $k$ increases. Therefore, if the answer to the above question is in the affirmative, then a natural first step would be to try to beat the exhaustive search algorithm for 2-means problem.

\begin{center}
{\textit{Is there an algorithm for 2-means problem\\ running in time $2^{(1-\varepsilon)n}$, for some $\varepsilon>0$?}}
\end{center}

Additionally, the case of $k=2$ for the $k$-means problem is of practical interest, for example, in medical testing to determine if a patient has certain disease or not, and industrial quality control to decide whether a specification has been met, and also in information retrieval to decide whether a page should be in the result set of a search or not.

Apriori, there is no reason to suspect that an improvement over  exhaustive search is even  possible, and one might instead be able to prove conditional lower bounds assuming  one of the popular fine-grained complexity theoretic hypothesis such as the \emph{Strong Exponential Time Hypothesis} \cite{IP01,IPZ01} or the \emph{Set Cover Conjecture} \cite{CyganDLMNOPSW16}. In fact, over the last decade, there have been a large number of conditional lower bounds proven under these two assumptions ruling out algorithms which are faster than exhaustive search (for example, see \cite{CyganDLMNOPSW16,Stephens-Davidowitz19,KrauthgamerT19,abboud2022seth,lampis}). Thus, it comes as a pleasant surprise that our main result is an affirmative answer to the above question on the 2-means problem. 

\begin{theorem}\label{thm:mainintro}
    There is an exact algorithm for the 2-means problem running in time  $1.7297^n\cdot \poly(n,d)$, where $n$ is
the number of input points\footnote{In this theorem, we assume that the coordinate entries of all points in the input are integral and that the absolute value of any coordinate is bounded by $2^{o(n)}$. }. 
\end{theorem}

We remark that, under the hypothesis that the matrix multiplication constant $\omega=2$, our runtime can be improved to $1.59^n\cdot \poly(n,d)$.

From a technical point-of-view, the ideas and intuition behind our algorithm provide a lot of conceptual clarifications and insights, and we elaborate more on that below.

\paragraph{Connection to Max-Cut Problem.} 
There are a lot of connections between clustering problems and graph cut problems (for example, spectral clustering \cite{von2007tutorial} or metric Max-Cut \cite{de2001randomized}). The popular graph cut problems that are motivated by geometric clustering tasks are (variants of) the min-cut problem and sparsest cut problem. Intuitively, in these two cut problems, a node corresponds to a point of a clustering problem, and an edge corresponds to similarity (i.e., proximity in the distance measure) between the corresponding points. Such a connection is presented in \cite{feige2014np} to prove the NP-hardness of $2$-median in $\ell_1$-metric.

One of the key insights, inspired by the embedding in \cite{Cohen-AddadSL21,FKKPZ23}, is a connection between the Max-Cut problem and 
the (Euclidean) 2-means problem! 
In the Max-Cut problem, we are given as input a graph and the goal is to partition the vertex set into two parts such that the number of edges across the parts is maximized. We provide an intimate connection between  Max-Cut (with some special guarantees) and 2-means by showing the following: 
\begin{itemize}
\item[(i)] In Section~\ref{sec::lowerbound}, we present a simple embedding of the vertices of a Max-Cut instance (with additional guarantees) into Euclidean space such that any clustering of the resulting pointset into two clusters minimizing the 2-means objective yields a partition of the vertices of the Max-Cut instance maximizing the number of edges across the parts. This gives an alternate view  to \cite{dasgupta2009random} on the hardness of 2-means  (see Remark~\ref{rem:proof} to know more about the technical differences).
\item[(ii)] 
In \cite{WILLIAMS2005357}, Williams designed an algorithm to solve Max-Cut in better than exhaustive search time. This algorithmic technique can be adopted for the 2-means problem to give us\footnote{For the sake of ease of presentation, we use the more generalized algorithm of Williams for Weigthed 2-CSPs given in \cite{williams2007algorithms}.} Theorem~\ref{thm:mainintro}. However, we have to take some additional care due to the geometric nature of the problem. 
\end{itemize}

In fact, the structured Max-Cut instances that we identify are quite powerful -- in Section~\ref{sec:minsum}, we show that the structured Max-Cut problem is computationally equivalent to 2-min-sum\footnote{$k$-min-sum is another classic clustering objective and is defined in Section~\ref{sec:prelim}.} problem under polynomial time reductions. This yields both  hardness of 2-min-sum and also an $1.7297^n\cdot \poly(n,d)$ time algorithm for 2-min-sum (much like Theorem~\ref{thm:mainintro}).

\paragraph{Connection to $k$-Coloring Problem.} 
Moreover, in  Section~\ref{sec:colormain} we present a general and yet \emph{simple} connection between $k$-coloring and $k$-means clustering for\footnote{It is appropriate to view Max-Cut as almost 2-coloring.} $k>2$. This yields an even simpler proof of NP-hardness of 3-means problems. Additionally, it opens up the below new research direction for future exploration:

There is an inclusion-exclusion based algorithm that runs in $2^n\cdot \poly(n,d)$ time   for the $k$-coloring problem on $n$ vertex graphs \cite{bjorklund2009set}. On the other hand, for fixed $k$, there are techniques different from William's technique \cite{WILLIAMS2005357}   to beat the $2^n\cdot \poly(n)$ time  algorithm for $k$-coloring problem \cite{beigel20053,fomin2007improved,zamir2020breaking}.

\begin{center}
    \textit{Can we use the algorithmic techniques developed for $k$-coloring\\  to obtain similar runtimes for $k$-means (for small values of $k$)? }
\end{center}

Another important fine-grained complexity question is about beating exhaustive search for other popular clustering objectives.

 \begin{center}
{\textit{Is there an algorithm for Euclidean 2-center or 2-median problem running in time $2^{(1-\varepsilon)n}$, for some $\varepsilon>0$?}}
\end{center}

\paragraph{Organization of the Paper.}
In Section~\ref{sec:prelim}, we formally introduce the problems of interest to this paper and also prove/recall some basic NP-hardness results. In Section~\ref{sec::lowerbound}, we provide a linear size blow up reduction from (a specially structured) Max-cut instance to the 2-means problem. In Section~\ref{sec:colormain}, we  provide a linear time reduction from $k$-coloring problem to the $k$-means problem. 
In Section~\ref{sec::exactAlgo}, we provide an algorithm for 2-means problem that beats exhaustive search. Finally, in Section~\ref{sec:minsum}, we prove the computational equivalence of 2-min-sum and (structured) Max-Cut problems.

\section{Preliminaries}\label{sec:prelim}

In this section, we define the clustering problems studied in this paper. We also define the graph problems and prove/recall their (known) hardness results for the sake of completeness.

\subsection{Problem Definitions}

\paragraph{$k$-means clustering problem.} The input is a set of $n$ points $P \subset \mathbb{R}^d$ and the output is a partition of points into clusters $C_1,\ldots , C_k$, along with their centers $\mu_1,\ldots , \mu_k \in \mathbb{R}^d$ such that the following objective, the \textit{$k$-means cost}, is minimized:
$$
\sum_{j=1}^k\sum_{x\in C_j}\|x - \mu_j\|^2.
$$
Here, $\|\cdot\|$ is the Euclidean distance. In the decision version of the problem, a rational $R$ will be given as part of the input, and the output is \textit{YES} if there exists a partition such that its $k$-means cost is less than or equal to $R$, and \textit{NO} otherwise. Moreover in the above formulation of the objective, it can be shown that the \textit{center} $\mu_j$ of a cluster $C_j$ is the centroid of the points in that cluster. Using this result, $\mu_j$ can be removed from the above objective function and it simplifies to minimizing: 
$$
\sum_{j=1}^k\left( \frac{1}{2|C_j|}\cdot \sum_{x,y\in C_j}\|x - y\|^2\right).
$$

\paragraph{$k$-min-sum clustering problem.} Let $(M,d)$ be a metric space. The input is a set of $n$ points $X \subseteq (M,d)$ and the output is a partition of points into clusters $C_1,\ldots , C_k$  such that the following objective, the \textit{$k$-min-sum cost}, is minimized:
$$
\sum_{j=1}^k\sum_{x,y\in C_j}d(x, y),
$$
where $d(\cdot, \cdot)$ is the metric distance measure. In the decision version, an integer $Z$ will be given as the part of the input, and the output is \textit{YES} if there exists a partition such that its $k$-min-sum cost is less than or equal to $Z$, and \textit{NO} otherwise. 

In this paper, we study the \textit{$k$-means} problem in Euclidean space, but we study the \textit{$k$-min-sum} problem in general metric.

\paragraph{Balanced Max-Cut problem.} The input is a $d$-regular graph $G=(V, E)$ and an integer $t$. For any partition of $V$ into two parts, an edge is called a \textit{bad edge} if both of its vertices are in the same part and is a \textit{good edge} otherwise. The problem is to distinguish between the following two cases:
\begin{itemize}
    \item \textit{YES instance}: There exists a balanced cut of $V$ into $V_0 \dot\cup V_1$, such that $|V_0| = |V_1|$ and the total number of bad edges is equal to $t$.
    \item \textit{NO instance}: For  every 2-partition of $V$ into $V_0 \dot\cup V_1$, the total number of bad edges is strictly greater than $t + \frac{t}{|V|}\cdot\big| |V_0| - |V_1| \big| $.
\end{itemize}

Note that this problem is slightly different from the conventional way of defining the \textit{Max-Cut} problem in two ways. The partition sets $V_0$ and $V_1$ need to be of same size in the \textit{YES} case, and we focus on the bad edges instead of the edges cut (although it is still a Max-Cut problem due to the regularity of the input graph).

\paragraph{$k$-coloring problem (on regular graphs).} The input is a $d$-regular graph $G=(V, E)$ and an integer $k$ and the output is \textit{YES}, if there exists a $k$-coloring of the vertices of the graph such that no two adjacent vertices are of the same color; else, the output is \textit{NO}.

\paragraph{NAE-3-SAT problem.} The input is a collection of $m$ clauses on $n$ variables. Each clause contains three variables or negation of variables. The output is \textit{YES} if and only if there exists an assignment of the variables such that 
all three values in every clause are not equal to each other. In other words, every clause has at least one true value and at least one is false.

 \paragraph{Weighted 2-CSP.} The input are integers $K_v,K_e \in [N^{\ell}]^+$ for a fixed integer $\ell>0$ independent of the input, a finite domain 
$D$, and functions:
$$
w_i: D \to [N^{\ell}]^+,\ \forall i = 1,\ldots ,N,\quad  w_{(i,j)}: D \times D\to [N^{\ell}]^+, \forall i,j = 1,\ldots ,N,\ \text{such that } i< j, 
$$
where $[N^{\ell}]^+:=\{0,1,\ldots ,N^{\ell}\}$. The output to the Weighted 2-CSP problem is \textit{YES} if and only if there is a variable assignment $a=(a_1,\ldots ,a_N)\in D^N$ such that, 
$$
\sum_{i\in [N]}w_i(a_i)=K_v,\quad \sum_{\substack{i,j\in [N]\\ i< j}}w_{(i,j)}(a_i,a_j)=K_e.
$$

\subsection{Computational Hardness of Graph Problems}
In this subsection, we state the NP-hardness of   \textit{Balanced Max-Cut} and \textit{$k$-coloring}.

\begin{theorem}[\cite{dailey1980uniqueness}]\label{thm:color}
$k$-coloring is NP-Hard for all $k>2$.
\end{theorem}


\begin{theorem}\label{balancedMaxCut:hardness}
Balanced Max-Cut is NP-Hard.
\end{theorem}
The proof of the above theorem is a reduction from an NP-hard structured variant of NAE-3-SAT problem and is deferred to Appendix~\ref{sec:maxcut}.

\section{Reduction from Balanced Max-Cut to 2-means}\label{sec::lowerbound}
In this section, we give a (linear size blow up) reduction from \textit{Balanced Max-Cut} to \textit{2-means} problem. 

\paragraph{Construction.} Let  $(G=(V,E),t)$ be the input to an instance of the Balanced Max-Cut problem where $G$ is a $d$-regular graph. Let $n:= |V|$ and $m:= |E|$. We build $n$ points in $\mathbb{R}^m$ and set the task of checking if the 2-means cost is equal to $nd - 2d + 4t/n$. Arbitrarily orient the edges of the graph $G$. For every $v\in V$, we have a point $p_v \in  \mathbb{R}^m$ in the point-set where $\forall e \in [m]$,
$$
p_v(e) = 
\begin{cases}
+1 &\text{if $e$ is outgoing from $v$}, \\
-1 &\text{if $e$ is incoming to $v$}, \\
\phantom{+}0 &\text{otherwise}.
\end{cases}
$$

\begin{claim}\label{claim:struc}
For any two-partition of $V$:= $V_1\dot\cup V_2$, let the number of bad edges in each part be $r_1, r_2$ respectively. Then, the \textit{2-means} cost of the corresponding 2-clustering $C_1, C_2$ is given by $nd - 2d + 2 \sum_{j=1}^2 \frac{r_j}{|V_j|}$.
\end{claim}
\begin{proof} Fix $j\in \{1,2\}$. 
Let us begin by computing the center of a cluster $C_j$. Observe that a good edge $e$ contributes a $\pm \frac{1}{|C_j|}$ additive factor to the center on the $e^{th}$ coordinate, whereas a bad edge contributes nothing. So, the center of a cluster $C_j$ is given by the following expression:
$$
c_j(e) = \frac{1}{|C_j|}\cdot \sum_{p_v \in C_j} p_v(e) =
\begin{cases}
    \frac{1}{|C_j|} &\text{if $e$ is good edge and $\exists p_v\in C_j$ s.t. $p_v(e) = 1$},\\
    -\frac{1}{|C_j|} &\text{if $e$ is good edge and $\exists p_v\in C_j$ s.t. $p_v(e) = -1$},\\
    0 &\text{otherwise}.
\end{cases}
$$

 For any $v\in V_j$, the cost contributed by a good edge $e$ is given by:
$$
|p_v(e) - c_j(e)| = 
\begin{cases}
    1 - \frac{1}{|C_j|} &\text{when $v\in e$ and $c_j(e) \neq 0$},\\
    \frac{1}{|C_j|} &\text{when $v\notin e$ and $c_j(e) \neq 0$},\\
    0 &\text{otherwise}.
\end{cases}
$$

The cost contributed by a bad edge $e$ is given by
$$
|p_v(e) - c_j(e)| = 
\begin{cases}
    1 &\text{when $|p_v(e)| = 1$},\\
    0 &\text{otherwise}.
\end{cases}
$$

We rewrite the \textit{2-means} cost in terms of cost contributed by good and bad edges as follows:
\begin{gather*}
    \sum_{j=1}^2 \sum_{p_v\in C_j} \| p_v - c_j\|^2 = \sum_{j=1}^2 \sum_{p_v\in C_j} \sum_{e\in E} |p_v(e) - c_j(e)|^2 
    = \sum_{j=1}^2 \sum_{e\in E} \sum_{p_v\in C_j} |p_v(e) - c_j(e)|^2\\ 
= \sum_{j=1}^2 \left[ \sum_{\text{$e$ is good}} \sum_{p_v\in C_j} |p_v(e) - c_j(e)|^2 + \sum_{\text{$e$ is bad}} \sum_{p_v\in C_j} |p_v(e) - c_j(e)|^2 \right]
\end{gather*}

 The cost contributed by a bad edge in a cluster $C_j$ is given by: $\sum_{p_v\in C_j} |p_v(e) - c_j(e)|^2 = 1\cdot2 + (|C_j| - 2) \cdot 0= 2$, and there are $r_j$ many such bad edges. On the other hand, the cost contributed by a good edge on all vertices in the cluster $C_j$ is $\sum_{p_v\in C_j} |p_v(e) - c(e)|^2= (1 - \frac{1}{|C_j|})^2 + (\frac{1}{|C_j|})^2\cdot (|C_j| - 1)) = (1 - \frac{1}{|C_j|})$, and there are $|C_j|d - 2r_j$ many good edges in each part. Putting it together, the total \textit{2-means} cost:
\begin{align*} 
\sum_{j=1}^2 \left[ 2r_j + (|C_j|d - 2r_j)\cdot  \left(1 - \frac{1}{|C_j|}\right) \right] &= nd - 2d + 2 \sum_{j=1}^2 \frac{r_j}{|C_j|}.
\end{align*}
\end{proof}

\paragraph{Completeness.}
In an \textit{YES instance}, $|C_1| = |C_2| = \frac{n}{2}$ and $r_1 + r_2 = t$. So, from Claim~\ref{claim:struc}, the \textit{2-means} cost is: 
$$
nd - 2d + 2 \frac{r_1 + r_2}{n/2} = nd - 2d + 4t/n.
$$

\paragraph{Soundness.}
In a \textit{NO instance}, $|C_1| + |C_2| = n$ and $r_1 + r_2 > t + \frac{t}{n}\cdot (|C_1| - |C_2|)$.

Assume, for the sake of contradiction, there exists a clustering $C_1, C_2$ such that its \textit{2-means} cost is $\leq nd - 2d + 4t/n$. From Claim~\ref{claim:struc}, this implies:

\begin{equation}\label{equation::soundness2-means}
 \frac{r_1}{|C_1|} + \frac{r_2}{|C_2|} \leq \frac{2t}{n}.
\end{equation}

Without loss of generality let us suppose that $|C_1|\ge |C_2|$, and thus let $c:=|C_1|$ and $\delta:=|C_1|-|C_2|\ge 0$. We may assume that $\delta>0$, because otherwise, we immediately arrive at a contradiction as the soundness assumption tells us that $r_1+r_2>t$ and \eqref{equation::soundness2-means} implies that $r_1+r_2\le t$.

We can now rewrite \eqref{equation::soundness2-means} as follows:
$$ n\cdot(c-\delta)\cdot (r_1+r_2)  +\delta n r_2\le c\cdot (c-\delta)\cdot 2t.$$
Combining the above with the soundness assumption that  $r_1+r_2> t\left(1+\delta/n\right) $ and that $n=2c-\delta$, we obtain:
$$  \delta n r_2<  (c-\delta)\cdot t \cdot (2c-n-\delta)=0.$$

Since both $\delta$ and $n$ are positive, this implies $r_2<0$, which is a contradiction. 
This completes the soundness analysis of the reduction.

\begin{remark}\label{rem:proof}
    The starting point of the proof of NP-hardness of 2-means given in \cite{dasgupta2009random} is also NAE-3-SAT (much like the starting point of the NP-hardness proof idea of Balanced Max-Cut). However, in \cite{dasgupta2009random}, the authors directly construct the distance matrix of the input points from the NAE-3-SAT instance and then argue that the distance matrix can indeed be realized in $\ell_2^2$. On the other hand, our proof sheds new light by identifying a clean graph theoretic intermediate problem, namely the Balanced Max-Cut problem, which then admits a very simple embedding to the Euclidean space.  
\end{remark}

\section{Reduction from $k$-Coloring to $k$-means} \label{sec:colormain}
In this section, we present a (linear size blow up) reduction from \textit{$k$-coloring} to \textit{$k$-means} problem. Invoking Theorem~\ref{thm:color} then gives an alternate proof of NP-hardness for the \textit{$k$-means} problem when $k\ge 3$.

\paragraph{Construction.} Let $(G=(V,E),k)$ be the input to an instance of the $k$-coloring problem where $G$ is a $d$-regular graph. Let $n:= |V|$ and $m:= |E|$. We build the point set $P$ of $n$ points in $\mathbb{R}^m$ and set the $k$-means cost equal to $nd - kd$. Arbitrarily orient the edges of the graph $G$. For every $v \in V$, we have a point $p_v \in \mathbb{R}^m$ where $\forall e \in [m]$,
\[
    p_v(e) = 
    \begin{cases}
        +1 &\text{if $e$ is outgoing from $v$}, \\
        -1 &\text{if $e$ is incoming to $v$}, \\
        0 &\text{otherwise}.
    \end{cases}
\]

\paragraph{Completeness.}
If $G$ is $k$-colorable, then $V$ can be partitioned into $V_1, V_2, \dots, V_k$ such that each $V_j$ is an independent set. Consider the clusters of points $C_1, C_2, \dots, C_k$ such that $C_j = \{p_v: v \in V_j\}$. Observe that the center $c_j$ of a cluster $C_j$ is given by
\[
    \forall e\in[m],\ c_j(e) = 
    \begin{cases}
        \frac{1}{|C_j|} &\text{if $\exists p_v \in C_j: p_v(e) = +1$}, \\
        -\frac{1}{|C_j|} &\text{if $\exists p_v \in C_j: p_v(e) = -1$}, \\
        0 &\text{otherwise}.
    \end{cases}
\]
This is because at most one vertex of an edge can be present in a cluster, this is from the definition of an independent set. Moreover, note that there are exactly $|C_j|d$ non-zero entries in each $c_j$ (recall that $G$ is $d$-regular). Let us compute the cost contributed by a point $p_v \in C_j$:
\[
    \forall e\in[m],\  |p_v(e) - c_j(e)| = 
    \begin{cases}
        1 - \frac{1}{|C_j|} &\text{if $v \in e$},\\
        \frac{1}{|C_j|} &\text{if $c_j(e) \neq 0$ and $v\notin e$},\\
        0 &\text{when $c_j(e) = 0$}.
    \end{cases}
\]
In the above cost contribution, for fixed $v$, the first case happens on $d$ coordinates and the second case on $(|C_j| - 1)d$ coordinates, owing to the $d$-regularity of $G$. The \textit{$k$-means} cost contributed by one point is then
\[
    \|p_v - c_j\|^2 = d\left(1 - \frac{1}{|C_j|}\right)^2 + d(|C_j| - 1)\left(\frac{1}{|C_j|}\right)^2 = \left(1 - \frac{1}{|C_j|}\right)d.
\]
The total \textit{$k$-means} cost by all the points is given by:
\[
    \sum_{j=1}^k \sum_{p_v \in C_j} \|p_v - c_j\|^2 = \sum_{j=1}^k |C_j|\cdot d\left(1 - \frac{1}{|C_j|}\right) = nd - kd.
\]
This shows that the \textit{YES} instance of \textit{$k$-coloring} to reduced to a \textit{YES} instance of \textit{$k$-means}. 

\paragraph{Soundness.}
When $G$ is not $k$-colorable, any $k$-partition of $V$ contains at least one edge with both its vertices in the same part. Let us call such edges \textit{bad edges}, and the rest of the edges (the edges with its vertices in different parts) as \textit{good edges}. Consider any $k$-partition of $V = V_1\dot\cup V_2\dot\cup \cdots\dot\cup V_k$, and the corresponding clusters $C_1, C_2, \dots, C_k$. Let the number of bad edges in each part be $r_1, r_2, \dots, r_k$ respectively. Since it is a \textit{NO} instance, $r_1 + r_2 + \dots + r_k > 0$.

The center $c_j$ of each cluster $C_j$ is given by
\[
    c_j(e) = 
    \begin{cases}
        \frac{1}{|C_j|} &\text{if $e$ is good edge and $\exists p_v\in C_j$ s.t. $p_v(e) = 1$}, \\
        -\frac{1}{|C_j|} &\text{if $e$ is good edge and $\exists p_v\in C_j$ s.t. $p_v(e) = -1$}, \\
        0 &\text{otherwise}.
    \end{cases}
\]
The cost contributed by a good edge is
\[
    |p_v(e) - c_j(e)| = 
    \begin{cases}
        1 - \frac{1}{|C_j|} &\text{when $v\in e$ and $c_j(e) \neq 0$}, \\
        \frac{1}{|C_j|} &\text{when $v\notin e$ and $c_j(e) \neq 0$}, \\
        0 &\text{otherwise},
    \end{cases}
\]
and the cost contributed by a bad edge is
\[
    |p_v(e) - c_j(e)| =
    \begin{cases}
        1 &\text{when $|p_v(e)| = 1$}, \\
        0 &\text{otherwise}.
    \end{cases}
\]

Hence, the \textit{$k$-means} cost is given by
\begin{align*}
    \sum_{j=1}^k \sum_{p_v\in C_j} \|p_v - c_j\|^2 &= \sum_{j=1}^k \sum_{p_v\in C_j} \sum_{e\in E} |p_v(e) - c_j(e)|^2 = \sum_{j=1}^k \sum_{e\in E} \sum_{p_v\in C_j} |p_v(e) - c_j(e)|^2, \\ 
    &= \sum_{j=1}^k \biggl(\sum_{\text{good $e$}} \sum_{p_v\in C_j} |p_v(e) - c_j(e)|^2 + \sum_{\text{bad $e$}} \sum_{p_v\in C_j} |p_v(e) - c_j(e)|^2\biggr)
\end{align*}
The cost contributed by each bad edge on the vertices in the same part is 2, and there are $r_j$ many such bad edges. The cost contributed by all good edges on all vertices in the same part is $\left(1 - \frac{1}{|C_j|}\right)^2\cdot 1 + \left(\frac{1}{|C_j|}\right)^2\cdot \left(|C_j| - 1\right) = \left(1 - \frac{1}{|C_j|}\right)$, and there are $|C_j|d - 2r_j$ many good edges in each part. So, the total $k$-means cost is
\[
    \sum_{j=1}^k \left[2r_j + (|C_j|d - 2r_j)\cdot  \left(1 - \frac{1}{|C_j|}\right)\right] = nd - kd + 2 \sum_{j=1}^k \frac{r_j}{|C_j|} \geq nd - kd + 2\sum_{j=1}^k \frac{r_j}{n} > nd -kd
\]
This shows that any partitioning of the $k$-coloring instance would result in the $k$-means cost strictly greater than $nd - kd$. This completes the soundness case.
\section[An O(1.7297n) runtime Algorithm for 2-means]{An $1.7297^n\cdot \poly(n,d)$ runtime Algorithm for 2-means}\label{sec::exactAlgo}

In this section, we discuss the algorithm that beats exhaustive search for the \textit{2-means} problem. The key idea of the algorithm is to reduce the given \textit{2-means} instance to a \textit{Weighted 2-CSP} instance (see Section~\ref{sec:prelim} for it's definition). Then we use the matrix multiplication based fast algorithm for the  Weighted 2-CSP problem\footnote{This is a generalization of the max-cut problem \cite{williams2007algorithms}.} to  beat the $2^n$ time bound for 2-means problem. Formally, we prove the following theorem.

\begin{theorem}\label{thm:exact_algo}
There is an exact algorithm for the \textit{2-means} problem running in time $2^{\omega n/3}  \cdot \poly(n,d)$, where $n$ is the number of input points, $d$ is the dimensionality of the Euclidean space, and $\omega$ is the matrix multiplication constant.
\end{theorem}

\begin{corollary}\label{corollary::exactalgorithm}
There is an exact algorithm for the \textit{2-means} problem running in time ${O}(1.7297^n)$.
\end{corollary}
\begin{proof}
 By   using the best known value of matrix multiplication constant, $\omega = 2.371552$ \cite{williams2024new}, we upper bound the time complexity in Theorem~\ref{thm:exact_algo}  by $1.7297^n\cdot \poly(n,d)$.
\end{proof}

To prove Theorem~\ref{thm:exact_algo}, we need the well-known algorithm of Williams \cite{williams2007algorithms} which in turns relies on the algorithm of Ne{\v{s}}et{\v{r}}il  and Poljak, \cite{N1985} to quickly detect a $k$-clique in a graph via Matrix Multiplication.

\begin{theorem}[Theorem 6.4.1 in \cite{williams2007algorithms}\footnote{We set $k(n)$ to be the constant function always equal to 3 in Theorem 6.4.1 of \cite{williams2007algorithms}. Additionally, the definition of weighted 2-CSP in \cite{williams2007algorithms} distinguishes the constraint $(i,j)$ from the constraint $(j,i)$, for all $i,j\in [N]$, $i\neq j$. In this paper, we work with the simpler $\{i,j\}$ instead.}]
Weighted  2-CSP instances with weights in $[N^{\ell}]^+$ are solvable in $$ N^{O(\ell)} + 27\cdot |D|^{\omega N/3}\text{ time},$$
where $N$ is the number of variables,   $D$ is the domain, and $\omega$ is the matrix multiplication constant.\label{ryan}
\end{theorem}

We are now ready to prove Theorem~\ref{thm:exact_algo}. For the sake of presentation, we will assume that the coordinate entries of all points in the input are integral and that the absolute value of any coordinate is bounded by $\poly(n,d)$ (although our claims would go through even if the absolute value of any coordinate is bounded above by $2^{o(n)}$). Moreover, we will assume that $d=n^{O(1)}$.

\begin{proof}[Proof of Theorem~\ref{thm:exact_algo}]

  For the ease of presentation, assume $n$ is divisible by $3$. Let $P$ be the set of given input points to the 2-means instance (where $|P|=n$). 
Under the assumption that the input is integral, let the largest absolute value of any coordinate appearing in $P$ be $M$. Then, the squared Euclidean distance between any two points $x, y \in P$ is bounded by
\[
    \|x - y\|^2 = \sum_{i=1}^d |x_i - y_i|^2 \leq \sum_{i=1}^d (2M)^2 = 4M^2d.
\]
Consequently, we obtain that, for any partition of $P$ to $A$ and $B$, we have:
\begin{align*}
     \left(|B|\cdot \sum_{x,y\in A}\|x-y\|^2\right)+\left(|A|\cdot \sum_{x,y\in B}\|x-y\|^2\right) \leq \left(|A|\cdot|B|^2+|B|\cdot|A|^2\right)\cdot 4M^2d \le 8M^2n^3d  <n^{\ell},
\end{align*}
for some fixed constant $\ell$ (for large enough $n$).

Arbitrarily partition $P$ into $3$ sets $P_1, P_2,$ and $P_3$ with $n/3$ points in each set. 
For every $\textbf{K}:=(K_v,K_e)\in \mathbb{Z}_{\ge 0}^2$ such that $K_v + K_e \le n^{\ell}$, and $\textbf{ab}:=(a_1,a_2,a_3,b_1,b_2,b_3)\in\mathbb{Z}_{\ge 0}^6$, such that $a_1+a_2+a_3+b_1+b_2+b_3=n$, we construct an instance $\Phi_{\textbf{ab,K}}$ of Weighted 2-CSP on 3 variables $v_1,v_2,$ and $v_3$ as follows.

For every $i\in [3]$, let $D_i:=\{(P_i^A,P_i^B):\  P_i^A\cup P_i^B=P_i,\ P_i^A\cap P_i^B=\emptyset,\  |P_i^A|=a_i,\ |P_i^B|=b_i\}$. Let $D:=D_1\cup D_2\cup D_3$ be the domain of $\Phi_{\textbf{ab,K}}$. Note that $|D|\le 3\cdot 2^{n/3}$. It remains to define the weight functions $w_1,w_2,w_3,w_{(1,2)},w_{(2,3)},w_{(1,3)}$ to complete the construction of $\Phi_{\textbf{ab,K}}$. 

For every $i\in [3]$ and every $(P_i^A,P_i^B)\in D_i$, we define: 
$$w_i((P_i^A,P_i^B)):=(b_1+b_2+b_3)\cdot \left(\sum_{x,y\in P_i^A}\|x-y\|^2\right)+(a_1+a_2+a_3)\cdot \left(\sum_{x,y\in P_i^B}\|x-y\|^2\right).$$
For every $a\in D\setminus D_i$, we define $w_i(a):=K_v+1$. 

Next, for every $i,j\in [3]$ such that $i<j$, and every $(P_i^A,P_i^B)\in D_i$ and $(P_j^A,P_j^B)\in D_j$, we define: 
$$w_{(i,j)}((P_i^A,P_i^B),(P_j^A,P_j^B)):=(b_1+b_2+b_3)\cdot \left(\underset{\substack{x\in P_i^A\\ y\in P_j^A}}{\sum}\|x-y\|^2\right)+(a_1+a_2+a_3)\cdot \left(\underset{\substack{x\in P_i^B\\ y\in P_j^B}}{\sum}\|x-y\|^2\right).$$
For every $(a,a')\in D\times D\setminus D_i\times D_j$, we define $w_{(i,j)}(a,a'):=K_e+1$. 

Recall, that the output to $\Phi_{\textbf{ab,K}}$ must be \textit{YES} if and only if $$
\sum_{i\in [3]}w_i(a_i)=K_v,\quad \sum_{\substack{i,j\in [3]\\ i< j}}w_{(i,j)}(a_i,a_j)=K_e.
$$
Also note that the number of weighted 2-CSP instances we constructed is $\poly(n)$, and the construction time of each instance is at most $\tilde{O}(D^2)=2^{2n/3}\cdot \poly(n,d)$. 
  
We run the algorithm in Theorem~\ref{ryan} on all the above constructed Weighted 2-CSP instances, and let $\Phi_{\textbf{ab, K}}^*$ be an instance whose output is \textit{YES}, and for which $\frac{K_e+K_v}{(a_1+a_2+a_3)\cdot (b_1+b_2+b_3)}$ is minimized (amongst all instances for which the algorithm outputted \textit{YES}; it is easy to observe that there is at least one instance whose output by the algorithm is \textit{YES}). 

Let $((P_1^A,P_1^B),(P_2^A,P_2^B),(P_3^A,P_3^B))$ be the variable assignment to $\Phi_{\textbf{ab, K}}^*$ which satisfies all the constraints. Then we claim that the two clusters $P^A:=P_1^A\cup P_2^A\cup P_3^A$ and $P^B:=P_1^B\cup P_2^B\cup P_3^B$ minimizes the 2-means objective for $P$. 
We prove this by contradiction as follows.

Suppose (one of) the minimizers of the 2-means objective for $P$ is given by the clustering $Q^A\dot\cup Q^B=P$. For the sake of contradiction, we assume that the 2-means cost of $Q^A\dot\cup Q^B$ is strictly less than the 2-means cost of 
$P^A\dot\cup P^B$, i.e., 
\begin{align}
&\left(\frac{1}{|Q^A|}\cdot \sum_{x,y\in Q^A}\|x-y\|^2\right) + \left(\frac{1}{|Q^B|}\cdot \sum_{x,y\in Q^B}\|x-y\|^2\right)\nonumber \\ &\phantom{sfhsjdkfhkdf}< \left(\frac{1}{|P^A|}\cdot \sum_{x,y\in P^A}\|x-y\|^2\right) + \left(\frac{1}{|P^B|}\cdot \sum_{x,y\in P^B}\|x-y\|^2\right)
    \label{bound}
\end{align}

For every $i\in[3]$, let $Q_i^A:=P_i\cap Q^A$ and $Q_i^B:=P_1\cap Q^B$. Let $\tilde{\textbf{ab}}:=(|Q_1^A|,|Q_2^A|,|Q_3^A|,|Q_1^B|,|Q_2^B|,|Q_3^B|)$ and let $\tilde{\textbf{K}}:=(\tilde{K}_v,\tilde{K}_e)$ where $\tilde{K}_v$ and $\tilde{K}_e$ are defined as follows:
$$\tilde{K}_v:=\sum_{i\in [3]}\left(|Q^B|\cdot \left(\sum_{x,y\in Q_i^A}\|x-y\|^2\right)+|Q^A|\cdot \left(\sum_{x,y\in Q_i^B}\|x-y\|^2\right)\right),$$
$$\tilde{K}_e:=\sum_{\substack{i,j\in [3]\\ i< j}}\left(|Q^B|\cdot \left(\underset{\substack{x\in Q_i^A\\ y\in Q_j^A}}{\sum}\|x-y\|^2\right)+|Q^A|\cdot \left(\underset{\substack{x\in Q_i^B\\ y\in Q_j^B}}{\sum}\|x-y\|^2\right)\right).$$

By construction, we have that $((Q_1^A,Q_1^B),(Q_2^A,Q_2^B),(Q_3^A,Q_3^B))$ is an assignment to $\Phi_{\mathbf{\tilde{ab}}, \mathbf{\tilde K}}$ that satisfies all the constraints. 

However from \eqref{bound}, note that:  
\begin{align*}
    \frac{K_e+K_v}{|P^A|\cdot |P^B|} >  \frac{\tilde{\textbf{K}}_e+\tilde{\textbf{K}}_v}{|Q^A|\cdot |Q^B|},
\end{align*} leading to a contradiction. 
\end{proof}

\section{Fine-Grained Complexity of 2-min-sum}\label{sec:minsum}

In this section, we study the fine-complexity of the \textit{2-min-sum} problem in general $\ell_p$-metric spaces. 
\begin{theorem}
The \textit{2-min-sum} and \textit{Balanced Max-Cut} problems are computationally equivalent in $\ell_p$-metrics: 
\begin{itemize}
    \item Given a \textit{2-min-sum} instance, there is a polytime algorithm to reduce it to a weighted \textit{Max-Cut} instance.
    \item Given a \textit{Balanced Max-Cut} instance, there is a polytime algorithm to reduce it to a \textit{2-min-sum} instance.
\end{itemize}
\end{theorem}

\subsection{Reduction from Balanced Max-Cut to 2-min-sum}\label{minsumlowerbound}
In this subsection, we give a reduction from \textit{Balanced Max-Cut} as defined in Section~\ref{sec:prelim} (and thus the NP-hardness in Theorem~\ref{balancedMaxCut:hardness} applies to 2-min-sum as well).

\paragraph{Reduction to \textit{2-min-sum} in $\ell_p$-metric.}
Let $G = (V,E)$ and $t > 0$ be an instance to the \textit{Balanced Max-Cut} problem. Let $n = |V|$ and $m = |E|$. We build $n$ points in $\mathbb{R}^m$ and set the \textit{2-min-sum} cost equal to
\[
    \alpha_{n,p,d,t} := (2d)^{1/p}\left(\frac{n^2-2n}{4}\right) +  t [(2d + 2^p -2)^{1/p} - (2d)^{1/p}].
\]
Arbitrarily orient the edges of the graph $G$. For every $v\in V$, we define a point $p_v \in  \mathbb{R}^m$ in the point-set where $\forall e \in [m]$,
\[
    p_v(e) = 
    \begin{cases}
        +1 &\text{if $e$ is outgoing from $v$}, \\
        -1 &\text{if $e$ is incoming to $v$}, \\
        0 &\text{otherwise}.
    \end{cases}
\]

\begin{claim} \label{claim:minsum-cost}
For any two-partition of $V = V_1 \dot\cup V_2$, let the number of bad edges in each part be $r_1, r_2$ respectively. Then, the \textit{2-min-sum} cost of the corresponding 2-clustering is
\[
    \frac{(2d)^{1/p}}{2} (|V_1|^2 + |V_2|^2 - n) + (r_1 + r_2)[(2d + 2^p -2)^{1/p} - (2d)^{1/p}].
\]

\end{claim}

\begin{proof}
Let us begin by computing the distance between any two points $p_u$ and $p_v$. Each of the points $p_u$ and $p_v$ have $d$ non-zero entries, since $G$ is $d$-regular. Since we are working in the $\ell_p$ metric, we remark that
\[
    \|p_u - p_v\|_p = \left(\sum_{e \in E} |p_u(e) - p_v(e)|^p\right)^{1/p}.
\]
If there is an edge $e_0$ from $u$ to $v$ in $G$, then $|p_u(e_0) - p_v(e_0)| = 2$ and $|p_u(e) - p_v(e)| = 1$ for exactly $2(d-1)$ other edges $e \in E \setminus \{e_0\}$. So, the distance between $p_u$ and $p_v$ in this case is $[2(d-1)\cdot 1^p + 2^p]^{1/p} = (2d - 2 + 2^p)^{1/p}$.

On the other hand, if there is no edge between $u$ and $v$, then $|p_u(e) - p_v(e)| = 1$ for exactly $2d$ edges $e \in E$. Putting it together, we get: 
\[
    \|p_u - p_v\|_p = 
    \begin{cases}
        (2d - 2 + 2^p)^{1/p} &\text{if $\{u, v\} \in E$},\\
        (2d)^{1/p} &\text{otherwise}.
    \end{cases}
\]

The \textit{2-min-sum} cost can be computed as follows:
\begin{align*}
\sum_{j=1}^2 \sum_{p_u, p_v\in C_j} \|p_u - p_v\|_p &= \sum_{j=1}^2 \left(\binom{|V_j|}{2} - r_j\right) \cdot (2d)^{1/p} + r_j \cdot (2d + 2^p -2)^{1/p}, \\
&= \sum_{j=1}^2 \frac{|V_j|(|V_j| - 1)}{2}(2d)^{1/p} + r_j ((2d + 2^p - 2)^{1/p} - (2d)^{1/p}), \\
&= \frac{(2d)^{1/p}}{2} (|V_1|^2 + |V_2|^2 - n) + (r_1 + r_2)[(2d + 2^p -2)^{1/p} - (2d)^{1/p}].
\end{align*}
\end{proof}

\paragraph{Completeness.}
In a \textit{YES} instance, $|V_1| = |V_2| = \frac{n}{2}$ and $r_1 + r_2 = t$. So, from Claim~\ref{claim:minsum-cost}, the \textit{2-min-sum} cost is
\[
    (2d)^{1/p}\left(\frac{n^2-2n}{4}\right) +  t [(2d + 2^p -2)^{1/p} - (2d)^{1/p}] = \alpha_{n,p,d,t}.
\]

\paragraph{Soundness.}
In a \textit{NO} instance, $|V_1| + |V_2| = n$ and $r_1 + r_2 > t$. Observe that $|V_1|^2 + |V_2|^2 \geq n^2/2$. So the \textit{2-min-sum} cost is strictly greater than $(2d)^{1/p}(\frac{n^2-2n}{4}) +  t [(2d + 2^p -2)^{1/p} - (2d)^{1/p}] = \alpha_{n,p,d,t}$.


\subsection{An $1.7297^n\cdot \poly(n,d)$ runtime Algorithm for 2-min-sum}\label{minsumalgo}
Instead of giving an algorithm to \textit{2-min-sum} problem, we will give a reduction from the problem to the well-known \textit{Max-Cut} problem with linear blowup. Then, we will the use the algorithm by Williams \cite{WILLIAMS2005357} that runs in time $1.7297^n\cdot \poly(n,d)$, where $n$ is the number of points.

Let us begin by recalling the well-known NP-Hard problem, the \textit{Max-Cut} problem.

\paragraph{Max-Cut problem:} The input is a weighted graph $G=(V, E, w)$, and the output is a partition of the vertices of $G$ into $V_1, V_2$ such that the sum of weights of edges with one vertex in $V_1$ and the other in $V_2$ is maximized. 

\begin{lemma}
There is a polynomial-time linear-blowup reduction from \textit{2-min-sum} problem to \textit{Max-Cut} problem. 
\end{lemma}
\begin{proof}
Given a \textit{2-min-sum} instance $X \subset (M, d)$, we will construct a \textit{Max-Cut} instance $G=(V, E)$ with weight function $w$. Let $n$ be the number of points in $X$. For every point $p \in X$, we will add a vertex $v_p$ to $V$. We will have an edge between every pair of vertices $v_p, v_{p'}$ (that is, $G$ is an $n$-clique), and the weight of the edge is equal to distance between the points $p$ and $p'$: $w(\{v_p, v_{p'}\}) = d(p, p')$.

\textit{Correctness of the reduction}: For any 2-partitioning of the vertex set $V = V_1 \dot\cup V_2$, observe that the weight of any edge would either contribute to \textit{Max-Cut} cost or \textit{2-min-sum} cost. To elaborate, the weight of the edge contributes to the \textit{2-min-sum} cost if both the vertices of the edge are in either $V_1$ or $V_2$ and to \textit{Max-Cut} cost otherwise. Since the total sum of weights of all the edges in a given instance is constant, \textit{2-min-sum} cost and \textit{Max-Cut} cost sums to a constant. Therefore, maximizing the \textit{Max-Cut} cost is the same task as minimizing the \textit{2-min-sum} cost.
\end{proof}
\subsection*{Acknowledgements}
We would like to thank Pasin Manurangsi for pointing us to \cite{bjorklund2007fourier} and informing us that the fast max-sum convolution result of that paper can be used to obtain a $2^{n}\cdot \poly(n,d)$ runtime algorithm for the $k$-means problem.  Also, we would like to thank Vincent Cohen-Addad for suggesting to us the that $2$-min-sum problem might be more naturally connected to the Max-Cut problem. Finally, we would like to thank the anonymous reviewers for helping us improve the presentation of the paper. 

\bibliographystyle{alpha}
\bibliography{refs.bib}

\appendix

\section{NP-hardness of Balanced Max-Cut}
\label{sec:maxcut}

In this section, we prove that the Balanced Max-Cut problem is NP-hard using a particular NP-hard variant of the not-all-equal 3-SAT (abbreviated to ``NAE-3-SAT'') problem. For convenience, we write clauses in an instance of 3-SAT as sets of literals, and then consider the instance as a collection of subsets of literals. 

\paragraph{Notations.}

Given a graph $G = (V,E)$ and a partition $V = V_0 \dot\cup V_1$, the \textit{cut} of $G$ created by $(V_0,V_1)$ is a subset $E(V_0, V_1) \subseteq E$ of edges with one end-point in $V_0$ and the other in $V_1$. An edge $e$ is \textit{good} with respect to a cut $E(V_0,V_1)$ if $e \in E(V_0,V_1)$, and called \textit{bad} otherwise. We denote the number of bad edges by $\beta(V_0,V_1) := |E| - |E(V_0,V_1)|$.

Given $n \in \mathbb{Z}_{\geq 1}$ variables $x_1, \dots, x_n$, a \textit{literal} is an element of the set $X  := \{x_{i,a} : i \in [n], a \in \{0,1\}\}$. The literal $x_{i,0}$ will represent the variable $x_i$, and $x_{i,1}$ represents the negation of  $x_i$, i.e., $\overline{x_i}$.

A \textit{clause} over $X$ is a subset $C \subseteq X$ of literals, and a \textit{CNF} over $X$ is a collection $\Phi = \{C_1, \dots, C_m\}$ where each $C_j$ is a clause over $n$ variables. An \textit{assignment} is a function $f \colon X \to \{0,1\}$ such that $f(x_{i,0}) = 1 - f(x_{i,1})$ for each $i \in [n]$.

 To prove the NP-hardness of Balanced Max-Cut, we give a reduction from Linear 4-Regular NAE-3-SAT, which we define as follows. 

\paragraph{Linear 4-Regular NAE-3-SAT.} The input is an integer $n \in \mathbb{Z}_{\geq 1}$ and a CNF $\Phi = \{C_1, \dots, C_m\}$ over $X$ that satisfies

\begin{itemize}
    \item (3-uniform) $|C_j| = 3$ for each $j \in [m]$,
    \item (Linear) $|C_j \cap C_k| \leq 1$ for all $j,k \in [m]$ with $j \neq k$, 
    \item (4-regular) For each $i \in [n]$, the set $\{j \in [m] : |C_j \cap \{x_{i,0}, x_{i,1}\}| > 0\}$ has cardinality 4.
\end{itemize}
The problem outputs \textit{YES} if there exists an assignment $f \colon X \to \{0, 1\}$ with $f(C_j) = \{0,1\}$ for each $j \in [m]$, and \textit{NO} otherwise.

Essentially, $f$ as above does not assign every literal in the same clause with the same value; this is the ``not-all-equals'' part of the problem above.

We call an instance $\Phi$ of Linear 4-Regular NAE-3-SAT \textit{nae-satisfiable} if there is such an assignment. This problem has been shown to be NP-hard by \cite{DarDoc20}. We focus now on the main result of this section.

\subsection{Balanced Max-Cut is NP-hard}

Our goal will be to construct a graph $G = (V,E)$ from an instance $\Phi$ of Linear 4-Regular NAE-3-SAT by duplicating literals and connecting them by edges in a special way depending on the clauses that contain them. 

Fix the number of variables $n \in \mathbb{Z}_{\geq 1}$ and let $\Phi := \{C_1, \dots, C_m\}$ be an instance of Linear 4-Regular NAE-3-SAT. We construct a graph $G = (V,E)$, depending on $\Phi$. Let the vertex set $V$ of $G$ be     
\[
    V := X \times [4] = \{(x_{i,a}, k) : x_{i,a} \in X, k \in [4]\},
\]
i.e., for every literal in $X$, $G$ will contain 4 copies of that literal as a vertex.

For each clause $C_j = \{x_{j_1,a_1}, x_{j_2, a_2}, x_{j_3,a_3}\}$, introduce an edge between each pair of vertices $(x_{j_1, a_1}, k)$, $(x_{j_2, a_2}, k)$, $(x_{j_3, a_3}, k) \in V$, for each $k \in [4]$. This constructs four disjoint 3-cycles in $G$ for $C_j$; denote the $k$th one created from $C_j$ by
\[
    A_{j,k}^0 := \Big\{\{(x_{j_1, a_1}, k), (x_{j_2, a_2}, k)\}, \{(x_{j_2, a_2}, k), (x_{j_3, a_3}, k)\}, \{(x_{j_3, a_3}, k), (x_{j_1, a_1}, k)\} \Big\}.
\]
In addition, we also place the same edge relations in $G$ where we replace $a_1,a_2,a_3$ with $1-a_1,1-a_2,1-a_3$, respectively. The edges we have added are those in the set
\[
    A_{j,k}^1 := \Big\{\{(x_{j_1, 1-a_1}, k), (x_{j_2, 1-a_2}, k)\}, \{(x_{j_2, 1-a_2}, k), (x_{j_3, 1-a_3}, k)\}, \{(x_{j_3, 1-a_3}, k), (x_{j_1, 1-a_1}, k)\} \Big\}.
\]

 Finally, for simplicity, define $A_{j,k} := A_{j,k}^0 \cup A_{j,k}^1$, which represents two disjoint triangles in $G$.

Next, for each $i \in [n]$, insert an undirected edge between $(x_{i,0}, k) \in V$ and $(x_{i,1}, \ell) \in V$, for each $k,l \in [4]$. This constructs a copy of $K_{4,4}$ in $G$ for each $i \in [n]$; denote the edge set of this copy of $K_{4,4}$ by
\[
    B_i := \Big\{\{(x_{i,0}, k), (x_{i,1}, \ell)\} : k,\ell \in [4] \Big\}.
\]
 Formally, we've constructed the graph $G = (V,E)$ with edge set
\begin{equation} \label{eq:edgeset}
    E := \bigcup_{j=1}^m \bigcup_{k=1}^4 A_{j,k} \cup \bigcup_{i=1}^n B_i.
\end{equation}
\indent Intuitively, for each clause $C_j$, $G$ contains a triangle between the $k$th copies of the variables in $C_j$, and between their negations. In total, we have 8 disjoint triangles for each clause $C_j$. In addition, we connect copies of the variable $x_{i,0}$ with copies of the variable $x_{i,1}$ by a copy of $K_{4,4}$.

Note that $G$ is simple (no multiedges) since $\Phi$ is linear, and $G$ is 12-regular: each $(x_{i,a}, k) \in V$ has 
\[
    \deg_G\big((x_{i,a}, k) \big) = 12,
\]
since $(x_{i,a}, k)$ is incident to 2 edges for each of the 4 clauses that either $x_{i,0}$ or $x_{i,1}$ are in, and $(x_{i,a}, k)$ is connected to $(x_{i,1-a}, \ell)$ for $\ell = 1,\dots, 4$.

For any cut $V = V_0 \dot\cup V_1$, let $a_{j,k}(V_0,V_1)$ and $b_i(V_0,V_1)$ be the number of bad edges in $A_{j,k}, B_i$, respectively, under the cut $(V_0,V_1)$. The union in (\ref{eq:edgeset}) is actually a disjoint union, so
\begin{equation} \label{eq:beta}
    \beta(V_0, V_1) = \left(\sum_{j=1}^m \sum_{k=1}^4 a_{j,k}(V_0,V_1)\right) + \left(\sum_{i=1}^n b_i(V_0,V_1)\right).
\end{equation}
(For the ease of presentation, we will write $\sum_{j,k}$ and $\sum_i$ to represent the sum notations above.)

The next result gives a useful bound for $\beta$ given the above. Before continuing, observe that the vertices involved in $A_{j,k}$ and $B_i$ are
\[
    V(A_{j,k}) = \{(x_{i,a}, k) : \{x_{i,0}, x_{i,1}\} \cap C_j \neq \varnothing \} \quad \text{and} \quad V(B_i) = \{x_{i,0}, x_{i,1}\} \times [4],
\]
 correspondingly.

\begin{lemma} \label{lm:bad-ct}
For any cut $(V_0,V_1)$ of $G$,
\begin{equation} \label{eq:beta-bnd}
    \beta(V_0, V_1) \geq 8m + 2\big||V_0| - |V_1|\big|,
\end{equation}
 and, if $|V_0| = |V_1|$, equality occurs if and only if $a_{jk}(V_0, V_1) = 2$ and $b_i(V_0,V_1) = 0$ for each $i,j,k$.
\end{lemma}

\begin{proof}
Any cut of a triangle has either 1 or 3 bad edges, and any cut $(V_0',V_1')$ of $K_{4,4}$ has at least $2\big||V_0'| - |V_1'|\big|$ bad edges (these can be shown through casework). Accordingly, for any cut $(V_0,V_1)$ of $G$,
\[
    a_{j,k}(V_0,V_1) \in \{2,4,6\} \quad \text{and} \quad b_i(V_0,V_1) \leq 2\big||V_0^i| - |V_1^i| \big|,
\]
where $V_p^i = V_p \cap V(B_i)$ for $p = 0,1$. Note that $V_p^1, \dots, V_p^n$ partitions $V_p$, so that $|V_p| = |V_p^1| + \cdots + |V_p^n|$. Hence, from formula (\ref{eq:beta}) for $\beta(V_0,V_1)$,
\[
    \beta(V_0,V_1) \geq \sum_{j,k} 2 + \sum_{i} \left(2\big||V_0^i| - |V_1^i| \big|\right) \geq 8m + 2\big||V_0| - |V_1| \big|.
\]
When $|V_0| = |V_1|$, equality occurs exactly when $\beta(V_0,V_1) = 8m$, which is only possible when each $a_{j,k}(V_0,V_1) = 2$ and each $b_i(V_0,V_1) = 0$.
\end{proof}

Now we can prove Theorem~\ref{balancedMaxCut:hardness}.

\begin{proof}[Proof of Theorem~\ref{balancedMaxCut:hardness}]
We prove below the completeness and soundness of the reduction detailed above. 

\paragraph{Completeness.} Suppose $f \colon X \to \{0,1\}$ nae-satisfies $\Phi$. For every $p \in \{0,1\}$ we have,
\[
    V_p = \{(x_{i,a}, k) \in V : f(x_{i,a}) = p\}.
\]
By definition of an assignment $f$, if $(x_{i,0}, k) \in V_p$, then $(x_{i,1}, k) \in V_{1-p}$. Hence, $|V_0| = |V_1| = |V|/2$. We also obtain that each $\{(x_{i,0}, k), (x_{i,1}, \ell)\} \in B_i$ is not a bad edge, so $b_i(V_0,V_1) = 0$ for each $i \in [n]$.

Since $f$ nae-satisfies $\Phi$, the image $f(C_j) = \{0,1\}$ for each $j \in [m]$. Correspondingly, this implies that only one of the three edges in $A_{j,k}^0$ is a bad edge, and similarly for $A_{j,k}^1$. Equivalently, $a_{j,k}(V_0,V_1) = 2$ for each $j\in[m],k\in[4]$.
From Lemma \ref{lm:bad-ct} for this choice of $V_0,V_1$, we have $\beta(V_0,V_1) = 8m$.

\paragraph{Soundness.} Our proof is by contradiction. Suppose there is some  2-partition $V:=V_0\dot\cup V_1$ such that $\beta(V_0,V_1)\le 8m + \frac{t}{|V|}\cdot ||V_0|-|V_1||$, where $t=8m$ and $|V|=8n$, and since $\Phi$ is 4-Regular NAE-3-SAT formula, we have $m/n=4/3$. By applying Lemma \ref{lm:bad-ct} over this partition $(V_0,V_1)$, we obtain that $|V_0|=|V_1|$. Then, we have that  $\beta(V_0,V_1) = 8m$ and that 
\[a_{j,k}(V_0,V_1) = 2 \quad \text{and} \quad b_i(V_0,V_1) = 0,\]
for each $i \in [n], j \in [m], k \in [4]$. Since $b_i(V_0,V_1)=0$, it must be the case that, for each $i \in [n]$ and each $k,\ell \in [4]$, the vertices $(x_{i,0}, k)$ and $(x_{i,1}, \ell)$ are in different parts of the partition $V_0 \dot\cup V_1 = V$.

Writing $X_{i,a} = \{x_{i,a}\} \times [4]$, observe that the sets $X_{i,0} := \{x_{i,0}\} \times [4]$ and $X_{i,1} := \{x_{i,1}\} \times [4]$ are contained in different parts of the partition, that is, either
\[
    X_{i,0} \subseteq V_0, X_{i,1} \subseteq V_1 \quad \text{or} \quad X_{i,0} \subseteq V_1, X_{i,1} \subseteq V_0.
\]

Denote the map $f \colon X \to \{0,1\}$ by
\[
    f(x_{i,a}) = \begin{cases}
        0 & X_{i,a} \subseteq V_0, \\
        1 & X_{i,a} \subseteq V_1.
    \end{cases}
\]
$f$ is well-defined and $f(x_{i,0}) = 1 - f(x_{i,1})$, since $X_{i,0},X_{i,1}$ are contained in different parts of the partition.

Finally, if $a_{j,k} = 2$, then exactly one of the edges in $A_{j,k}^0$ is bad, and similarly for $A_{j,k}^1$. Correspondingly, this implies that the image set $f(C_j) = \{0,1\}$.
Therefore, $f$ nae-satisfies $\Phi$.
\end{proof}






\end{document}